\documentclass[letterpaper, 10 pt, conference]{ieeeconf}  

\IEEEoverridecommandlockouts                              

\usepackage{amsmath,amssymb,amsfonts} 
\usepackage{graphicx} 
\usepackage{siunitx} 
\usepackage{subfig} 
\newtheorem{thm}{Theorem}[section]
\newtheorem{cor}{Corollary}
\newtheorem{rem}{Remark}

\newcommand*\diff{\mathop{}\!\mathrm{d}}

\title{\LARGE \bf
Steady-state nonlinearity of open-loop reset systems
}

\author{M. B. Kaczmarek$^{1}$, X. Zhang$^{1}$ and  S. H. HosseinNia$^{1}$
\thanks{*This work was not supported by any organization}
\thanks{$^{1}$Authors are with Faculty of Mechanical, Maritime and Materials Engineering (3mE),
        Delft University of Technology, Delft, The Netherlands
        {\tt\small m.b.kaczmarek@tudelft.nl}}%
}

\begin{document}

\maketitle
\thispagestyle{empty}
\pagestyle{empty}

\begin{abstract}

 In this paper, we introduce a new representation for open-loop reset systems. We show that at steady-state a reset integrator can be modelled as a parallel interconnection of the base-linear system and piece-wise constant nonlinearity. For sinusoidal input signals, this nonlinearity takes a form of a square wave. Subsequently, we show how the behaviour of a general open-loop reset system is related to the nonlinearity of a reset integrator. The proposed approach simplifies the analysis of reset elements in the frequency domain and provides new insights into the behaviour of reset control systems.

\end{abstract}

\section{Introduction}
\label{sec:Introduction}

A reset element is a linear time-invariant system whose states, or a subset of states, reset to values defined by a reset law if certain conditions are satisfied \cite{BanosBook}. It has been proven that reset systems can overcome limitations of linear controllers \cite{Zheng2000, Beker2001}. Examples of applications of reset in various fields like process control or networked systems can be found in textbooks \cite{BanosBook,GuoBook,asrcsBook}. Moreover, reset elements have been successfully applied to control precision positioning systems \cite{Heertjes2015, VanLoon2017, CgLp, Heertjes2016, Hazeleger2016, Saikumar2019, Chen2019}. 

One of the reasons why the reset control systems draw so much attention, especially for industrial applications, is the fact that they can be designed using a modified frequency-domain loop-shaping procedure. The steady-state behaviour of nonlinear systems may be, in certain cases, described using the \textit{describing functions} (DF) \cite{Rijlaarsdam2017}. For reset system the DF were first derived in \cite{Guo2009}, and later extended to the \textit{ Higher-order Sinusoidal-input Describing function} (HOSIDF) in \cite{Saikumar2021}.

With HOSIDF, the nonlinearity of a system is represented by describing harmonics of the output signal to a sinusoidal input at a certain frequency. While this method can be used for tuning the controllers, the influence of the higher harmonics can not be easily interpreted. 

In this paper, we introduce a new representation of the steady-state behaviour of reset systems with sinusoidal inputs. First, we show that a reset integrator can be modelled at steady state as a parallel interconnection of the base-linear systems and piece-wise constant non-linearity. For sinusoidal input signals,  this nonlinear component of the output takes a form of a square wave. 

Any general open-loop reset system can be represented as a feedback system built around a reset integrator. Using this fact, we show how the nonlinearity of any reset element in open-loop is related to the behaviour of a reset integrator.

The representation introduced in this paper provides a clear interpretation of HOSIDF of open-loop reset systems. 
\section{Background}
\label{sec:Background}
\subsection{Notation}
$I$ and $0$ denote here square identity matrix and zero matrix of appropriate size respectively. $J_{m,n}$ denotes a matrix of ones.

\subsection{Reset control systems}

Consider a reset element 
\begin{equation}\label{eq:R}
R:
\begin{cases}
\dot x_r(t) = A_r x_r(t) + B_r u_r(t), & \text{if } u_r \neq 0\\
x_r(t^+) = A_{\rho} x_r(t),  & \text{if } u_r = 0\\
y_r = C_r x_r(t) + D_r u_r(t)
\end{cases},
\end{equation}
where  $x_r(t^+) =  \lim_{\epsilon \rightarrow 0^+} x(t+\epsilon)$, $x_r \in \mathbb{R}^{m}$ is the state of $R$, $u_r \in \mathbb{R}^{1}$ is the input of $R$, $y_r \in \mathbb{R}^{1}$ is the output of $R$ and {$A_r$, $B_r$, $A_{\rho}$, $C_r$, $D_r$} are constant matrices of appropriate dimensions.

The linear system described with ($A_r,B_r,C_r,D_r$) is referred to by the term \textit{Base linear system} (BLS) and describes dynamics of $R$ in absence of reset. 

The \textit{linear reset law} $x_r(t^+) = A_{\rho}x_r(t)$ describes the change of state that occurs at \textit{reset instants} $t_k, k = 1,2,\dots$, that is when the \textit{reset condition} $u_r = 0$ is satisfied. Alternative reset laws and conditions \cite{asrcsBook} are not considered in this work.



Since reset systems are a special case of hybrid systems, pathological behaviours like beating, deadlock and Zeno behaviour may occur \cite{Haddad2014}. In practice, existence and uniqueness of the solution are assured by time-regularization \cite{Nesic2008,Zaccarian2005}. Time-regularization is a modification of reset system, such that reset instants happen only if a minimum time between resets $\Delta_m > 0$ has lapsed. Any discrete-time implementation inherently features time regularization with $\Delta_m$ equal to the sampling time \cite{Heertjes2016}. In remainder of this paper it is assumed that solutions of $R$ are well defined \cite{BanosBook}.

\subsection{Reset elements}

For illustration purpose, let introduce the \textit{First order reset element} (FORE) and the \textit{Second order reset element} (SORE). FORE is a reset element \eqref{eq:R} with
\begin{align*}
   A_r &= -\omega_r, & B_r &= \omega_r,&  C_r &= 1,&  D_r &= 0.
\end{align*}
The state-space matrices for SORE are
\begin{align*}
   A_r &= \begin{bmatrix} 0 & 1\\ -\omega_r^2 & -2\beta_r\omega_r \end{bmatrix}, & B_r &= \begin{bmatrix} 0 \\ \omega_r^2\end{bmatrix},\\  C_r &= \begin{bmatrix} 1 & 0 \end{bmatrix},&  D_r &= 0.
\end{align*}
Here, we choose $\omega_r = 100$, $\beta_r = 0.1$ and $A_\rho = 0$. 
\section{Nonlinearity of reset systems}
\label{sec:Derivations}

To analyse the steady-state behaviour of a system, we require it to be \textit{uniformly convergent} \cite{Pavlov2006a,Pavlov2007}. In this way, we guarantee the existence of a steady-state solution for a system driven by periodic inputs. This condition is satisfied by \eqref{eq:R} for a class of sinusoidal inputs if $|\lambda(A_\rho e^{A_r \delta})|<1, \forall \delta\in\mathbb{R}^+$ \cite{Guo2009}. In \cite{Dastjerdi2020b}, this property has been studied in closed-loop, related to the $H_\beta$ condition \cite{Beker2004} and proven for a wider class of input signals.
\begin{thm}[Reset integrator]\label{thm:Rint}
The steady-state state response of uniformly convergent reset integrator, i.e.\  reset system \eqref{eq:R} with $A_r = 0, B_r = I, C_r = I, D_r = 0$, to an input signal $u_r(t) = a \sin(\omega t), a \in \mathbb{R}^{m\times 1}$ is given by 
\begin{equation}
    x_r(t) = x_{bls}(t) + q_i(t),
\end{equation}
where $x_{bls}$ denotes the steady-state response of the base-linear system $(A_r,B_r,C_r,D_r)$ and $q_i$ is a square wave in phase with $u_r$,
with a mean value
\begin{equation}
\bar q_i = \frac{-a}{\omega}I
\end{equation}
and amplitude
\begin{equation}
\hat q_i = (I-A_\rho)(I + A_\rho)^{-1}\frac{a}{\omega}.
\end{equation}
\end{thm}

\begin{proof}
Consider dynamics of $q \triangleq x_r-x_{bls}$
\begin{equation}\label{eq:dq}
    \begin{cases}
    \dot q(t) = A_r q(t), & \text{if } u_r \neq 0\\
    q(t^+) = A_\rho q(t)+(A_\rho-I)x_{bls}(t), & \text{if } u_r = 0.
    \end{cases}
\end{equation}
It is clear that for a reset integrator $A_r  =0$ with any input $u_r$, $q(t)$ is piecewise constant, and jumps only at the reset instants.

 Between consecutive reset instants $t_k,t_{k+1}$ we have
\begin{equation}
    q(t) = q(t_k^+)+\int_{t_k}^t A_r q(\tau)\diff \tau, \text{ for } t\in(t_k, t_{k+1}).
\end{equation}

For an input signal $u_r(t) = a \sin(\omega t)$, we have reset instants $t_{k} = k\frac{\pi}{\omega}, k\in \mathbb{N}$. Signal $q$ specific for this case is denoted with subscript $_i$. From the second equation of \eqref{eq:dq} 
\begin{equation}
\begin{split}
    q_i(t_k^+) &= 
    A_\rho q_i(t_{k-1}^+) + (A_\rho -I)\int_{t_{k-1}}^{t_{k}} B_r u_r(\tau) \diff \tau \\
    &= \begin{cases}
    A_\rho q_i(t_{2n}^+), & \text{if } k = 2n+1\\
    A_\rho q_i(t_{2n+1}^+) + (A_\rho -I)\frac{2}{\omega}a, & \text{if } k = 2n.
    \end{cases}
\end{split}    
\end{equation}
At the steady state, we have
\begin{equation}
    \begin{split}
        q_i(t_{2n}^+) &= q_i(t_{2n+2}^+)\\
        &= A_\rho q_i(t_{2n+1}^+)+(A_\rho -I)\frac{2}{\omega}a \\
        &= A_\rho^2 q_i(t_{2n}^+)+(A_\rho -I)\frac{2}{\omega}a.
    \end{split}
\end{equation}
After solving for $q_i(t_{2n}^+)$ we have
\begin{equation}
 q_i(t_{2n}^+) = (I-A_\rho^2)^{-1}(A_\rho - I)\frac{2}{\omega}a = -(I+A_\rho)^{-1}\frac{2}{\omega}a.
\end{equation}
For $q_i(t_{2n+1}^+)$ we have
\begin{equation}
    q_i(t_{2n+1}^+) = -A_\rho(I+A_\rho)^{-1}\frac{2}{\omega}a.
\end{equation}

Taking all this into consideration we have 
\begin{equation}
    q_i(t) = 
    \begin{cases}
    -(I+A_\rho)^{-1}\frac{2}{\omega}a, & \text{for } t \in [t_{2n},t_{2n+1})\\
    -A_\rho(I+A_\rho)^{-1}\frac{2}{\omega}a, & \text{for } t \in [t_{2n+1},t_{2n+2}).
    \end{cases}
\end{equation}

The mean value $\bar q_i$ and the peak amplitude $\hat q_i$ are given by
\begin{align}
    \bar q_i &= \frac{q(t_{2n}^+) + q(t_{2n+1}^+)}{2} = \frac{-a}{\omega}I\\
    \hat q_i &= \frac{q(t_{2n+1}^+) - q(t_{2n}^+)}{2} = (I-A_\rho)(I+A_\rho)^{-1}\frac{a}{\omega}.
\end{align}

So we can also write
\begin{equation}
    q_i(t) = 
    \begin{cases}
    \bar q_i + \hat q_i, & \text{for } t \in [t_{2n},t_{2n+1})\\
    \bar q_i - \hat q_i, & \text{for } t \in [t_{2n+1},t_{2n+2}).
    \end{cases}
\end{equation}
\end{proof}

\begin{rem}
Signal $q(t)$ represents nonlinearity added to the base linear system by the reset actions.
\end{rem}

Steady-state response of a reset integrator, divided into the linear and nonlinear components, is presented in Fig. \ref{fig:TimeInt}.

\begin{figure} 
    \centering
  \subfloat[Integrator ($A_r = 0, B_r = 1, C_r = 1, D_r = 0, A_\rho = 0$) \label{fig:TimeInt}]{%
       \includegraphics[width=0.45\textwidth]{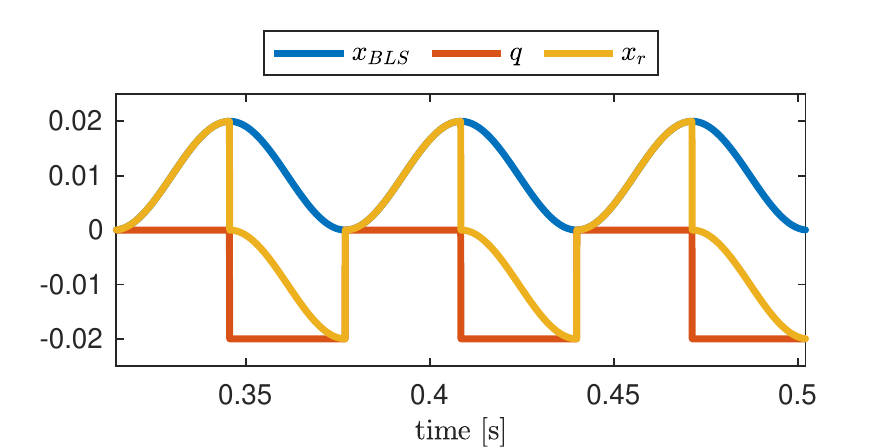}}
\\
  \subfloat[FORE ($A_r = -100, B_r = 100, C_r = 1, D_r = 0, A_\rho = 0$) \label{fig:TimeFore}]{%
        \includegraphics[width=0.45\textwidth]{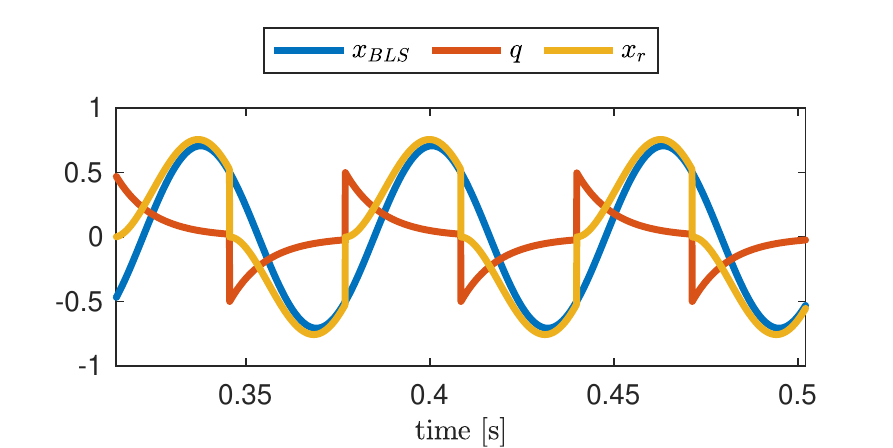}}
  \caption{Steady-state time-responses of reset elements to $u_r(t) = \sin(100 t)$ as a sum of linear and nonlinear contributions.}
  \label{fig:Time} 
\end{figure}

\begin{thm}[Reset state-space system]\label{thm:Rtot}
The steady-state state response of {uniformly convergent} reset system \eqref{eq:R} to an input signal $u_r(t) = b \sin(\omega t), b \in \mathbb{R}$ is given by 
\begin{align}
    x_{r} &= x_{bls} + q = x_{bls} + T_{q}\circledast q_i,\\
    T_{q}(s) &=\frac{x_r(s)}{q_i(s)}  = Q(sI-A_r)^{-1}s
\end{align}
where $x_{bls}$ denotes the steady-state response of the base-linear system $(A_r,B_r,C_r,D_r)$ and $q_i$ is the square wave introduced in Theorem \ref{thm:Rint}, representing the nonlinearity in a reset integrator with the same number of states and reset matrix as in \eqref{eq:R}, that is driven by input $a\sin(\omega t), a = bJ_{m,1}$.

The scaling of the magnitude of the nonlinearity $Q\in \mathbb{R}^{m \times m}$ is given by
\begin{equation}
    \min_Q \left\lVert e^{A_r t}(I-A_\rho)x_{bls}(t_k)- e^{A_r t}A_\rho Q q^*(t_k)+Qq^*(t)\right\rVert_2^2,
\end{equation}
for $t\in (t_k,t_{k+1})$, where $t_k,t_{k+1}$ denote subsequent reset instants.
\end{thm}
\begin{proof}
Any state-space reset system can be represented in a block diagram form as a feedback system with a reset integrator and feedback gain $A_r$, see Fig. \ref{fig:Rss1}. Using the Theorem \ref{thm:Rint}, the reset integrator can be represented as a parallel interconnection of its base-linear system and a block generating the signal related to resets (Fig. \ref{fig:Rss2}). 

Reset instants of considered elements are depending only on the input signal. In consequence, reset instants for a reset integrator considered in Theorem \ref{thm:Rint} and a general reset state-space system are the same, if both systems are driven by sine waves of the same frequency and phase. However, in the state-space system, the signal entering the integrator $e_r$ is shifted with respect to the $u_r$. This leads to the change of the magnitude of the square wave, which is represented by the scaling matrix $Q$.

\begin{figure} 
    \centering
  \subfloat[\label{fig:Rss1}]{%
       \includegraphics[width=0.43\textwidth]{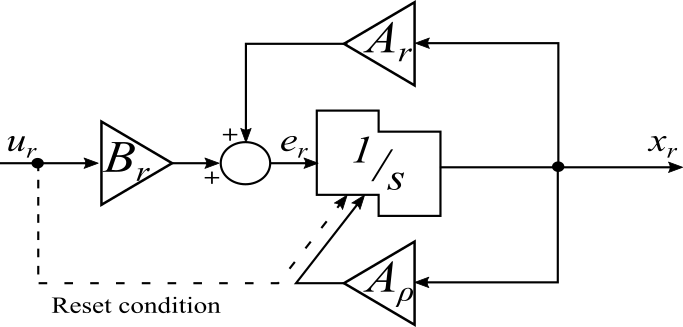}}
\hfill
  \subfloat[\label{fig:Rss2}]{%
        \includegraphics[width=0.43\textwidth]{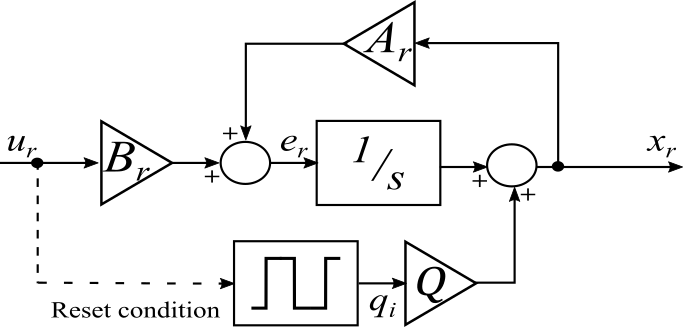}}
\hfill
  \subfloat[\label{fig:Rss3}]{%
        \includegraphics[width=0.43\textwidth]{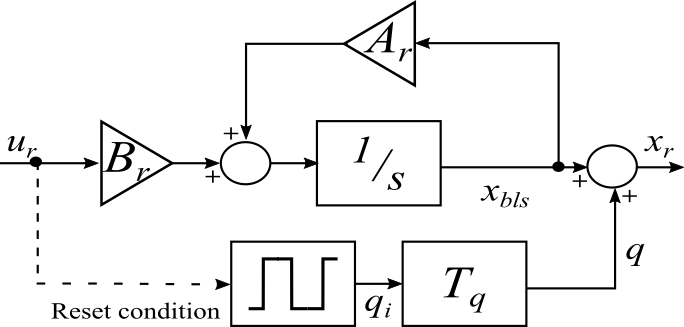}}
   \caption{Block diagram representation of a reset state-space system. a) standard form, b) reset integrator as a sum of linear and nonlinear components, c) entire state-space element as a sum of linear and nonlinear components. Dotted lines indicate signals triggering resets.}
  \label{fig:Rss} 
\end{figure}

The complete new representation of a reset element is shown in Fig. \ref{fig:Rss3}. It can be seen that the state of a reset element consists of the response of the base linear system $x_{bls}$ and a contribution due to the nonlinearity $q_i$. This contribution is denoted by $q$ and defined in \eqref{eq:dq}. By analysing the block diagrams, we have
\begin{align}
    x_r &= A_r\frac{I}{s}x_r + B_r\frac{I}{s} u_r + q,\\
    x_{r} &= x_{bls} + q = x_{bls} + T_{q}\circledast q_i,\\
    T_{q}(s) &=\frac{x_r(s)}{q_i(s)}  = Q(sI-A_r)^{-1}s
\end{align}

The linear transfer function $T_{q}$ defines the magnitude and shape of the nonlinear contribution. 

To find the amplitude of the nonlinear contribution we consider evolution of a reset system states between reset instants. We define $q(t) \triangleq Q q^*(t)$, where $q^*$ can be interpreted as a response of $T_q$ with $Q = I$ to the square wave signal $q_i$. 

Using that $x_r = x_{bls}+q$, for $t\in (t_k,t_{k+1})$ we have
\begin{align}
 e^{A_r(t-t_k)}x_{bls}(t_k) + q(t) = &e^{A_r(t-t_k)}x_{r}(t_k^+)\\
 e^{A_r(t-t_k)}x_{bls}(t_k) + Qq^*(t) = &e^{A_r(t-t_k)}A_\rho (x_{bls}(t_k)+Qq^*(t_k))\\
0 = e^{A_r t}(I-A_\rho)x_{bls}(t_k)- &e^{A_r t}A_\rho  Q q^*(t_k)+Qq^*(t).
\end{align}
{In consequence, $Q$ can be found as a solution for a minimization problem}
\begin{equation}
    \min_Q \left\lVert e^{A_r t}(I-A_\rho)x_{bls}(t_k)- e^{A_r t}A_\rho Q q^*(t_k)+Qq^*(t)\right\rVert_2^2.
\end{equation}
For first-order reset elements ($m=1$), it is sufficient to consider only the before and after-reset state. From \eqref{eq:dq} we have
\begin{align}
A_\rho \left(x_{bls}(t_k)+q(t_k)\right) &= x_{bls}(t_k) + q(t_k^+),\\
(A_\rho-I)x_{bls}(t_k) &= Q\left(q^*(t_k^+) - A_\rho q^*(t_k)\right),\\
Q = (A_\rho-I) x_{bls}&(t_k)\big(q^*(t_k^+) - A_\rho q^*(t_k)\big) ^{-1}.
\end{align}

To complete the calculation of $Q$ we need to find $x_{bls}(t_k)$ and $q^*(t)$.

The response of the base-linear system at a reset instant is given by
\begin{equation}
\begin{split}
    x_{bls}(t_k) &= e^{A_r t}x_{bls}(0) + \int_0^t e^{A_r(t-\tau)}B_r a \sin(\omega \tau)\diff \tau \\ &= \begin{cases} 
    ~~(A_r^2+\omega^2I)^{-1}B_r a\omega, \text{ for } k = 2n+1\\
    -(A_r^2+\omega^2I)^{-1}B_r a\omega, \text{ for } k = 2n+2.
    \end{cases}
\end{split}
\end{equation}
As we mentioned earlier, $q^*$ can be interpreted as a response of $T_q$ with $Q = I$ to $q_i$. $T_q$ can be represented in the state-space form by the quadruple (ABCD-matrices) $(A_r,A_r,I,I)$. Note, that $T_q$ has a direct feedthrough term. We calculate response of a linear system $T_q$ to a square wave $q_i$
\begin{equation}\label{eq:q*t}
    q^*(t) = e^{A_r t}q^*(0) + \int_0^t e^{A_r(t-\tau)}A_r q_i(\tau)\diff \tau + q_i(t).
\end{equation}
At the steady state, the term related to $q^*(0)$ disappears. $q_i$ consists of a constant and a variable component. The contribution due to the constant component is
\begin{equation}
    \int_0^t e^{A_r(t-\tau)}A_r \bar q_i\diff \tau = (e^{A_r t}-I)\bar q_i \approx -\bar q_i \text{ (for large $t$)}.
\end{equation}

Consider now $t \in [t_{2n}, t_{2n+1})$. The varying part $q_i(t) - \bar q_i$ is a zero-mean square wave with peak amplitude $\hat q_i$. We have
\begin{equation}
    \int_0^{t_{2n}} e^{A_r(t-\tau)}A_r (q_i(\tau) -\bar q_i)\diff \tau = 0
\end{equation}
The response for $t \in [t_{2n}, t_{2n+1})$ can be calculated as a response to step with magnitude $\hat q_i$, because the constant components $\bar q_i$ cancel out. For this, we define the state of $T_q$ to be $x_q$, such that $q(t)^* = Ix_q(t) + Iq_i(t)$. The initial condition is  $x_q(t_{2n})$. At the steady state we have $x_q(t_{2n}) = -x_q(t_{2n+1})$. By comparing these values we get
\begin{align}\label{eq:qtk}
    -e^{A_r \frac{\pi}{\omega}}x_q(t_{2n+1}) + (e^{A_r \frac{\pi}{\omega}}-I)\hat q = x_q(t_{2n+1})\\
    (e^{A_r \frac{\pi}{\omega}}+I)x_q(t_{2n+1}) = (e^{A_r \frac{\pi}{\omega}}-I)\hat q_i\\
    x_q(t_{2n+1}) = (e^{A_r \frac{\pi}{\omega}}+I)^{-1}(e^{A_r \frac{\pi}{\omega}}-I)\hat q_i\\
    q^*(t_{2n+1}) = \big((e^{A_r \frac{\pi}{\omega}}+I)^{-1}(e^{A_r \frac{\pi}{\omega}}-I) +I \big)\hat q_i
\end{align}
Following the step response logic, the after reset value $q^*(t_{2n+1}^+)$ is given by
\begin{align}\label{eq:qtk+}
    q^*(t_{2n+1}^+) = x_q(t_{2n+1})-I\hat q_i.
\end{align}

\end{proof}

To illustrate the Theorem \ref{thm:Rtot}, Fig. \ref{fig:TimeFore} shows a steady-state response of a FORE. The FORE can be represented as a feedback system around the reset integrator, whose response is presented in Fig. \ref{fig:TimeInt}. In consequence, the nonlinear components $q$ presented in both figures are related by $T_q$.

The following corollaries are a consequence of the fact that an ideal square wave can be represented as an infinite sum of sinusoidal waves
\begin{equation}
    q_i(t) = \bar q_i + \hat q_i \frac{4}{\pi}\left(\sin(\omega t) + \frac{1}{3}\sin(3\omega t)+\frac{1}{5}\sin(5\omega t)\dots \right).
\end{equation}

\begin{cor}\label{cor:qharmonics}
The $k$-th harmonic of the nonlinear contribution $q$ in the reset system \eqref{eq:R} is given  by 
\begin{equation}
    q_k(\omega) = 
    \begin{cases}
     \frac{4}{k\pi}C_rQ(jk\omega I - A_r)^{-1}jk\omega \hat q_i J_{m,1} & \text{for odd } k\\
     0 &\text{for even } k.
    \end{cases}
\end{equation}
\end{cor}

\begin{cor}\label{cor:HOSIDF}
The Higher-order Sinusoidal-input Describing function (HOSIDF) of the reset system \eqref{eq:R} is given  by 
\begin{equation}
    H_k(\omega) = 
    \begin{cases}
    C_r\left(j\omega I-A_r\right)^{-1}B_r+D_r + q_1(\omega) & \text{for } k=1\\
     q_k(\omega) & \text{for } k\geq 2
    \end{cases}
\end{equation}
\end{cor}
Corollaries \ref{cor:qharmonics} and \ref{cor:HOSIDF} give a clear insight into the behaviour of HOSIDF for reset elements.

\begin{rem}
For a reset integrator, the influence of a higher-order harmonic on the behaviour of the element is inversely proportional to the order of the harmonic.
\end{rem}

Nonlinearities of a general reset system and a corresponding reset integrator are related by a linear transfer function $T_q$, which depends on the state matrix $A_r$ of the base linear system of the system and the reset matrix $A_\rho$. In consequence, the influence of a particular harmonic can be amplified or diminished.

\begin{rem}
For a uniformly convergent reset element, the sum of HOSIDF is always finite, as they compose the nonlinearity $q$.
\end{rem}

\section{Examples}
\label{sec:Examples}




\begin{figure}[b]
  \centering
  \includegraphics[width=0.45\textwidth]{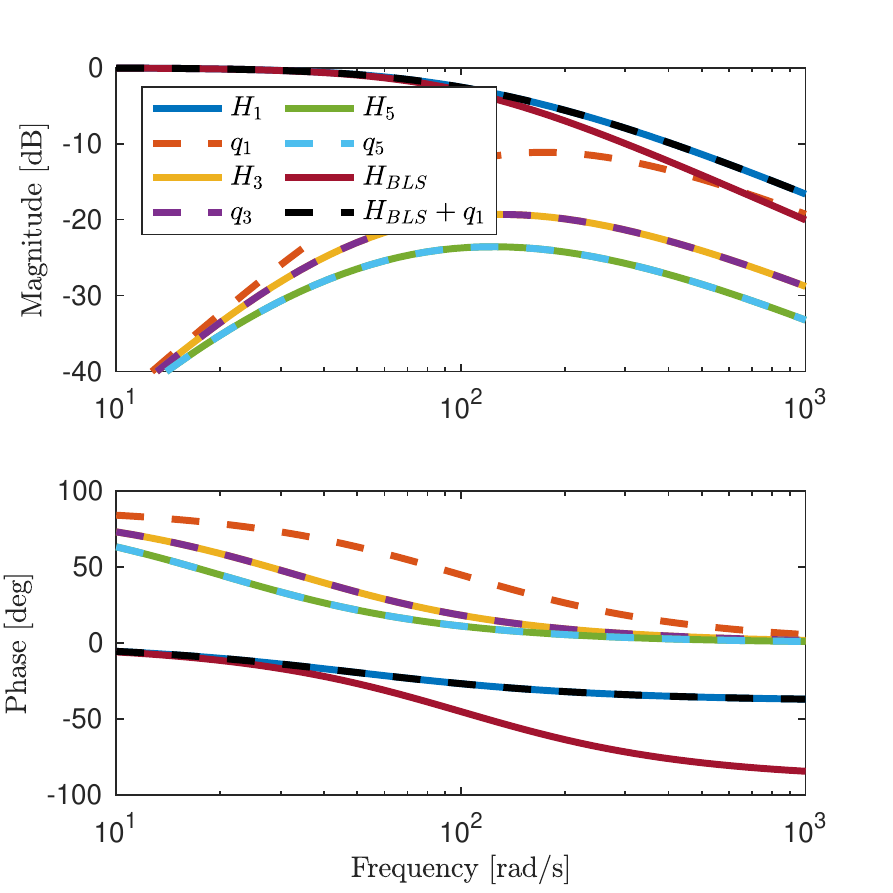}
  \caption{HOSIDF $H_j$ and harmonics of the nonlinearity $q_j$ for the FORE element.}
  \label{fig:hosidfFore}
\end{figure}

\begin{figure}[t]
  \centering
  \includegraphics[width=0.45\textwidth]{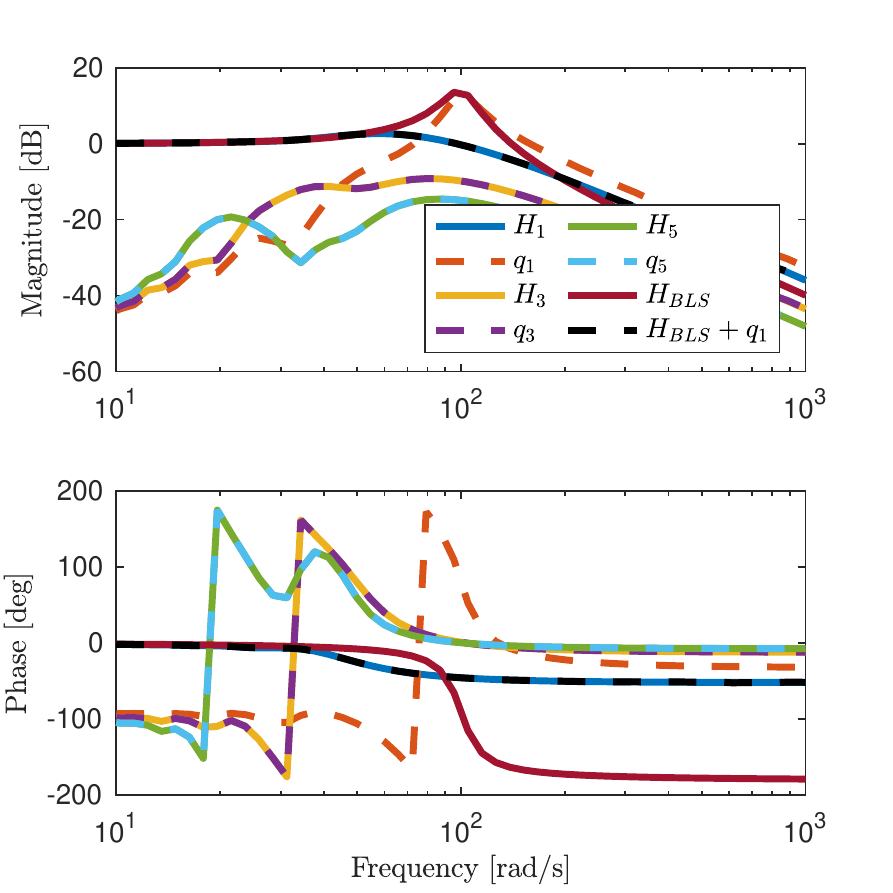}
  \caption{HOSIDF $H_j$ and harmonics of the nonlinearity  for the SORE element.}
  \label{fig:hosidfSore}
\end{figure}

In this section we will illustrate the equivalence of the introduced representation of reset systems with the standard models. To show that the introduced methodology represents the steady-state behaviour of a reset system correctly, we compare the HOSIDF derived in Corollaries \ref{cor:qharmonics} and \ref{cor:HOSIDF} with ones derived using the original method introduced by \cite{Guo2009} and \cite{Saikumar2021}.

Figures \ref{fig:hosidfFore} and \ref{fig:hosidfSore} compare the HOSIDF and the harmonics of the nonlinear component $q$ of FORE and SORE. The differences between the 1st-order describing function, 1st harmonic of the nonlinear contribution and the transfer function of the base-linear system can be clearly seen. The 1st-order HOSIDF is recreated by summing the transfer function of the BLS and the 1st harmonic of $q$. As expected, the higher-order harmonics of $q$ and HOSIDF are equivalent.




\section{Conclusion} \label{sec:Conclusion}
In this paper, we introduced a new representation for the steady-state responses of open-loop reset systems. It provides an intuitive explanation for  HOSIDF of reset systems and highlights the influence of nonlinearity on the behaviour of the reset system. The insights from this work can be applied for the shaping of nonlinearities of reset elements. While this paper considers only reset elements in open-loop, the same approach may be taken in the case of closed-loop systems if appropriate assumptions about the reset instants are made \cite{Saikumar2021}.

\addtolength{\textheight}{-12cm}   








\bibliographystyle{unsrt}
\bibliography{references}

\begin{thebibliography}{10}

\bibitem{BanosBook}
Antonio Barreiro and Alfonso B{\~{a}}nos.
\newblock {\em {Reset control systems}}.
\newblock Springer-Verlag, London, 2012.

\bibitem{Zheng2000}
Y~Zheng, Y~Chait, C~V Hollot, M~Steinbuch, and M~Norg.
\newblock {Experimental demonstration of reset control design}.
\newblock {\em Control Engineering Practice}, 8(2):113--120, 2 2000.

\bibitem{Beker2001}
O~Beker, C~V Hollot, and Y~Chait.
\newblock {Plant with integrator: An example of reset control overcoming
  limitations of linear feedback}.
\newblock {\em IEEE Transactions on Automatic Control}, 46(11):1797--1799, 11
  2001.

\bibitem{GuoBook}
Yuqian Guo, Lihua Xie, and Youyi Wang.
\newblock {\em {Analysis and design of reset control systems}}.
\newblock Institution of Engineering and Technology, 1 2016.

\bibitem{asrcsBook}
Christophe Prieur, Isabelle Queinnec, Sophie Tarbouriech, and Luca Zaccarian.
\newblock {Analysis and synthesis of reset control systems}.
\newblock {\em Foundations and Trends in Systems and Control}, 6(2-3):119--338,
  2018.

\bibitem{Heertjes2015}
M~F Heertjes, K~G~J Gruntjens, S~J L~M Van~Loon, N~Kontaras, and W~P M~H
  Heemels.
\newblock {Design of a variable gain integrator with reset}.
\newblock In {\em Proceedings of the American Control Conference}, volume
  2015-July, pages 2155--2160. Institute of Electrical and Electronics
  Engineers Inc., 7 2015.

\bibitem{VanLoon2017}
S~J L~M van Loon, K~G~J Gruntjens, M~F Heertjes, N~van~de Wouw, and W~P M~H
  Heemels.
\newblock {Frequency-domain tools for stability analysis of reset control
  systems}.
\newblock {\em Automatica}, 82:101--108, 8 2017.

\bibitem{CgLp}
N~Saikumar, R~K Sinha, and S~H HosseinNia.
\newblock {"Constant in Gain Lead in Phase” Element– Application in
  Precision Motion Control}.
\newblock {\em IEEE/ASME Transactions on Mechatronics}, 24(3):1176--1185, 6
  2019.

\bibitem{Heertjes2016}
M~F Heertjes, K~G~J Gruntjens, S~J L~M van Loon, N~van~de Wouw, and W~P M~H
  Heemels.
\newblock {Experimental Evaluation of Reset Control for Improved Stage
  Performance}.
\newblock {\em IFAC-PapersOnLine}, 49(13):93--98, 2016.

\bibitem{Hazeleger2016}
Leroy Hazeleger, Marcel Heertjes, and Henk Nijmeijer.
\newblock {Second-order reset elements for stage control design}.
\newblock In {\em Proceedings of the American Control Conference}, volume
  2016-July, pages 2643--2648. Institute of Electrical and Electronics
  Engineers Inc., 7 2016.

\bibitem{Saikumar2019}
N.~Saikumar, D.~Valerio, and S.~H. HosseinNia.
\newblock {Complex order control for improved loop-shaping in precision
  positioning}.
\newblock In {\em Proceedings of the IEEE 58th Conference on Decision and
  Control, CDC}, 7 2019.

\bibitem{Chen2019}
Linda Chen, Niranjan Saikumar, and S~Hassan HosseinNia.
\newblock {Development of Robust Fractional-Order Reset Control}.
\newblock {\em IEEE Transactions on Control Systems Technology}, pages 1--28,
  2019.

\bibitem{Rijlaarsdam2017}
David Rijlaarsdam, Pieter Nuij, Johan Schoukens, and Maarten Steinbuch.
\newblock {A comparative overview of frequency domain methods for nonlinear
  systems}.
\newblock {\em Mechatronics}, 42:11--24, 2017.

\bibitem{Guo2009}
Yuqian Guo, Yuoyi Wang, and Lihua Xie.
\newblock {Frequency-domain properties of reset systems with application in
  hard-disk-drive systems}.
\newblock {\em IEEE Transactions on Control Systems Technology},
  17(6):1446--1453, 2009.

\bibitem{Saikumar2021}
Niranjan Saikumar, Kars Heinen, and S.~Hassan HosseinNia.
\newblock {Loop-shaping for reset control systems: A higher-order
  sinusoidal-input describing functions approach}.
\newblock {\em Control Engineering Practice}, 111:104808, 6 2021.

\bibitem{Haddad2014}
Wassim~M Haddad, Vijay~Sekhar Chellaboina, and Sergey~G Nersesov.
\newblock {\em {Impulsive and hybrid dynamical systems: Stability,
  dissipativity, and control}}.
\newblock Princeton University Press, 9 2014.

\bibitem{Nesic2008}
Dragan Nesic, Luca Zaccarian, and Andrew~R Teel.
\newblock {Stability properties of reset systems}.
\newblock {\em Automatica}, 44(8):2019--2026, 2008.

\bibitem{Zaccarian2005}
Luca Zaccarian, Dragan Ne{\v{s}}i{\'{c}}, and Andrew~R Teel.
\newblock {First order reset elements and the Clegg integrator revisited}.
\newblock {\em Proceedings of the 2005, American Control Conference, 2005.},
  pages 563--568 vol. 1, 2005.

\bibitem{Pavlov2006a}
Alexey Pavlov, Nathan van~de Wouw, and Henk Nijmeijer.
\newblock {\em {Uniform Output Regulation of Nonlinear Systems}}.
\newblock Birkh{\"{a}}user Boston, 2006.

\bibitem{Pavlov2007}
Alexey Pavlov, Nathan van~de Wouw, and Henk Nijmeijer.
\newblock {Frequency Response Functions for Nonlinear Convergent Systems}.
\newblock {\em IEEE Transactions on Automatic Control}, 52(6):1159--1165, 6
  2007.

\bibitem{Dastjerdi2020b}
Ali~Ahmadi Dastjerdi, Alessandro Astolfi, Niranjan Saikumar, Nima Karbasizadeh,
  Duarte Valerio, and S.~Hassan HosseinNia.
\newblock {Closed-Loop Frequency Analysis of Reset Control Systems}.
\newblock page~1, 2020.

\bibitem{Beker2004}
O~Beker, C~V Hollot, Y~Chait, and H~Han.
\newblock {Fundamental properties of reset control systems}.
\newblock {\em Automatica}, 40(6):905--915, 6 2004.

\end{thebibliography}

\end{document}